\documentclass [12pt]{article}

\usepackage{amssymb,amsmath,amsthm}
\usepackage{graphicx}
\usepackage[T1]{fontenc}

\long\def\comment#1{{}}

\newtheorem{thm}{Theorem}[section]
\newtheorem{lemma}[thm]{Lemma}
\newtheorem{prop}[thm]{Proposition}
\newtheorem{cor}[thm]{Corollary}

\theoremstyle{definition}
\newtheorem{df}[thm]{Definition}
\newtheorem{ex}[thm]{Example}

\begin{document}

\title{{\bf The key properties of inconsistency for a triad in pairwise comparisons matrices}}

\author{W.W. Koczkodaj
\thanks{
Computer Science, Laurentian University,\;\;
Sudbury, Ontario P3E 2C6, Canada 
wkoczkodaj@cs.laurentian.ca 
} \\
\and
J. Szybowski
\thanks{
AGH University of Science and Technology,\; 
Faculty of Applied Mathematics,\;
al. Mickiewicza 30, 30-059 Krakow, Poland, szybowsk@agh.edu.pl}
\thanks{Research supported by the 
Polish Ministry of Science and Higher Education}
} 

\maketitle

\begin{abstract}
Processing information, acquired by subjective assessments, involves inconsistency analysis in most (if not all) applications of which some are of considerable importance at a national level (see, \cite{KKL2014})
A triad inconsistency axiomatization in pairwise comparisons was informally proposed in \cite{KS2014a}. 
This study, rectifies it by the use of the distance and theoretical proofs. Three key properties of the indicator are presented in this study and illustrated by several examples. \\


\end{abstract}

\section{Introduction-preliminaries}

The importance of pairwise comparisons and its practical usefulness has been recently evidenced by P. Faliszewski, L.A and E Hemaspaandras in one of the flagship ACM publication \cite{FHH2010}. The triad consistency 
was introduced in \cite{KS2015a}, extended and analyzed in \cite{DK1994,HK1996M-C,HK1996converg}. The presented concept of triads is similar to \cite{AM2005} and \cite{CRR2009}. However, this study is self-contained and does not require any listed reference to follow.

From the mathematical point of view, pairwise comparisons (PC) matrix $M$ represents strictly positive ratios of entities 
$E_1, E_2, \ldots, E_n$ for $i,j \in \{1,2,...n\}$

Assuming that  $m_{ik}, m_{kj}, m_{ij}$ represent ratios of entities , the consistency condition simply states that: $m_{ij}= m_{ik} m_{kj}$. If entities are somehow measurable, we have: 
 $$\frac{E_i}{E_j}=\frac{E_i}{E_k}\cdot\frac{E_k}{E_j}$$

\noindent and it is consistent with the common sense. Regretfully, most entities are highly subjective (e.g., public safety or public satisfaction) hence not easy to measure but we can relate them to each other. It means that our input (ratios) is less restrictive and applicable to situations where a common ``yardstick'' simply does not exist.

In the case of $3 \times  3$ PC matrix, three PC matrix elements above the main diagonal represent the entire PC matrix if it is reciprocal matrix.
In particular, two elements above the main diagonal are sufficient for the entire PC matrix reconstruction from the consistency condition if the PC matrix is consistent.
since $m_{1,3}=m_{1,2}\cdot m_{2,3}=m_{2,3}\cdot m_{1,2}$.
and this simple fact is now used in our axiomatization.

Let us revisit the axioms proposed in \cite{KS2014a} for the inconsistency indicator $ii$
(analyzed independently by Bozoki and the late Rapcsak in \cite{BR2008}) for a single triad $(x,y,z)$:

\begin{enumerate}
    \item $ii=0$ for $y=x*z$,
    \item $ii \in [0,1)$,
    \item for a consistent triad $ii(x,y,z)=0$ with $xz=y$, increasing or decreasing $x,y,z$ results in increasing $ii(x,y,z)$.
\end{enumerate}

The third axiom was vaguely worded and it has been found that the clarification is needed.
The proper mathematical analysis for a single triad is provided in this study. For it, a special metric $d$ is introduced. 
%
The mathematical properties of inconsistency indicators 
will be discussed in the next section. 

\section{The main results}


\begin{df}
A function $td:\ {\mathbb R}_+^3  \rightarrow [0,+\infty)$ is called {\em a triad 
deviation}, if there exists a metric 
$d:\ {\mathbb R}_+^2 \rightarrow [0,+\infty)$ such that $\forall x, y, z$ it holds 
\begin{equation}
 td(x,y,z)=d(xz,y). \label{td}
\end{equation}
\end{df}

\begin{df}
We say that a triad deviation $td$ satisfying (\ref{td}) is {\em induced by a metric $d$}.
\end{df}

\begin{flushleft}
Proposition 2.3 follows straight from the definition of a triad deviation.
\end{flushleft}

\begin{prop}
For all positive real numbers $a,b,c,d$ and $e$, a triad deviation $td$ satisfies the conditions
\begin{eqnarray}
td(a,b,c)=0 &\Leftrightarrow& ac=b \label{zero} \\
td(a,b,c)&=&td(b,ac,1) \label{comm}\\ 
td(a,de,c) &\leq& td(a,b,c)+td(d,b,e) \label{tri}
\end{eqnarray}
\end{prop}

As a consequence of (\ref{comm}), we have:

\begin{cor}
For all positive real numbers $a,b,c$
\begin{equation}
td(a,b,c)=td(c,b,a). \label{comm2}
\end{equation}
\end{cor}

\begin{prop}
For each function $td:\ {\mathbb R}_+^3  \rightarrow [0,+\infty)$ satisfying $(\ref{zero})-(\ref{tri})$ function $d_{td}:\ {\mathbb R}_+^2 \rightarrow [0,+\infty)$ given by 
\begin{equation}
d_{td}(x,y)=td(x,y,1) \label{ind_met}
\end{equation}
 is a metric. Moreover, the triad deviation $td$ is induced by the metric $d_{td}$.
\end{prop}

\begin{proof}
From $(\ref{zero})$, it follows that $d_{td}(x,y)=0\Leftrightarrow td(x,y,1)=0\Leftrightarrow x=y$.
Condition (\ref{comm}) implies that $d_{td}(y,x)=td(y,x,1)=td(x,y,1)=d_{td}(x,y)$.
Finally, from (\ref{tri}) we get $d_{td}(x,z)=td(x,z,1)\leq td(x,y,1)+td(z,y,1)=d_{td}(x,y)+d_{td}(z,y)=d_{td}(x,y)+d_{td}(y,z)$, which proves that $d_{td}$ is a metric. 

The former statement follows also from the sequence of equalities:
$$d_{td}(xz,y)=d_{td}(y,xz)=td(y,xz,1)=td(x,y,z).$$
\end{proof}

%
%
\begin{df}
We say that a metric $d_{td}$, satisfying (\ref{ind_met}), is {\em induced by a triad deviation $td$}.
\end{df}

\begin{df}
We say that a triad deviation $td$ is {\em bounded} if $\exists M>0 \; \forall a,b,c \in {\mathbb R}_+\ td(a,b,c) \leq M$.
\end{df}

\begin{df}
We say that a bounded triad deviation $ii$ is a {\em triad inconsistency indicator} if $\forall a,b,c \in {\mathbb R}_+\ ii(a,b,c) \leq 1$.
\end{df}

\begin{ex}
The function $DI:\ {\mathbb R}_+^3  \rightarrow \{0,1\}$ given by 
\[ DI(a,b,c)=
 \left\{
  \begin{array}{cl}
   0, & ac=b\\
   1, & \mbox{otherwise}
  \end{array}
 \right.
\]
 is a triad inconsistency indicator induced by a discrete metric.
\end{ex}

\begin{ex}
The function $EI:\ {\mathbb R}_+^3  \rightarrow [0,+\infty)$ given by $$EI(a,b,c)=|ac-b|$$ is an unbounded triad deviation induced by a Euclidean metric.
\end{ex}

\begin{ex}
The function $I_1:\ {\mathbb R}_+^3  \rightarrow [0,1)$ given by $$I_1(a,b,c)=\frac{|ac-b|}{1+|ac-b|}$$ is a triad inconsistency indicator induced by a metric $d_1$ given by the formula $d_1(x,y)=\frac{|x-y|}{1+|x-y|}$.
\end{ex}


The distance-based inconsistency indicator was proposed in \cite{Kocz93} and generalized in \cite{DK1994}:

$$Kii(a,b,c) = \min \left(|1-\frac{b}{ac}|,|1-\frac{ac}{b}|\right)$$ 

\noindent for all triads $(x,y,z)$ specified by the consistency condition.

\noindent It was simplified in \cite{KS2014a} to:

$$Kii(a,b,c) = 1-\min\left(\frac{b}{ac},\frac{ac}{b}\right),$$ 

\noindent which is equivalent to:

$$Kii(a,b,c)=1- e^{-\left|\ln\left (\frac{b}{ac}\right )\right |}.$$

\begin{lemma} \label{metric}
Function $d: \ {\mathbb R}_+^2 \rightarrow [0,1)$ given by formula
$$d(x,y)=1-\min\left(\frac{x}{y},\frac{y}{x}\right)$$
is a metric.
\end{lemma}

\begin{proof}
Evidently, $$d(x,y)=0 \Leftrightarrow \min\left(\frac{x}{y},\frac{y}{x}\right)=1 \Leftrightarrow x=y,$$
and $$d(x,y)=1-\min\left(\frac{x}{y},\frac{y}{x}\right)=1-\min\left(\frac{y}{x},\frac{x}{y}\right)=d(y,x).$$
For the proof of the triangle inequality
$$d(x,z)\leq d(x,y)+d(y,z)$$
notice, that it is equivalent to the inequality
\begin{equation}\label{triangle}
\min\left(\frac{x}{y},\frac{y}{x}\right)+\min\left(\frac{y}{z},\frac{z}{y}\right) \leq 1+\min\left(\frac{x}{z},\frac{z}{x}\right).
\end{equation}
Consider three cases:

$1^{o}\ x \leq y \leq z$

then 
$$(\ref{triangle}) \Leftrightarrow \frac{x}{y}+\frac{y}{z} \leq 1+ \frac{x}{z} \Leftrightarrow xz+y^2 \leq yz + xy \Leftrightarrow 0 \leq (y-x)(z-y),$$
which is evident to:

$2^{o}\ x \leq z \leq y$

then 
$$(\ref{triangle}) \Leftrightarrow \frac{x}{y}+\frac{z}{y} \leq 1+ \frac{x}{z} \Leftrightarrow xz+z^2 \leq yz + xy \Leftrightarrow 0 \leq (y-z)(x+z)$$

$3^{o}\ y \leq x \leq z$

then 
$$(\ref{triangle}) \Leftrightarrow \frac{y}{x}+\frac{y}{z} \leq 1+ \frac{x}{z} \Leftrightarrow yz+xy \leq xz + x^2 \Leftrightarrow 0 \leq (x-y)(x+z)$$

In all three cases, $(\ref{triangle})$ follows immediately from $(\ref{comm2})$.
\end{proof}

\begin{thm}
$Kii$ is a triad inconsistency indicator.
\end{thm}
\begin{proof}
$Kii$ is given by the formula: $Kii(a,b,c)=d(ac,b),$ where $d$ is a bounded metric from Lemma \ref{metric}.
\end{proof}

\begin{ex}
In \cite{PL2003}, a seemingly similar inconsistency index for triads was introduced:
$$PL(a,b,c)=\frac{b}{ac}+\frac{ac}{b}-2.$$
As it was shown in \cite{KuS}, the inconsistency indicator of a matrix based on $PL$ is not equivalent to $Kii$.

It is easy to see that $PL$ is not bounded, since $PL(1,n,1)=n+\frac{1}{n}-2$, which may be arbitrarily large.
Moreover, $PL$ is not even a triad deviation!  To see it, let us use $$a=1,\; b=3,\; c=5,\; d=1,\; e=2$$ and calculate 
$$PL(1,3,5)+PL(1,3,2)=\frac{13}{30}$$ and $$PL(1,2,5)=\frac{9}{10},$$ which contradicts (\ref{tri}).
\end{ex}

Finally, it should be acknowledged that \cite{KB1939} defined an $n \times n$ matrix $A$ with positive elements as {\em a pairwise comparisons matrix (a PC-matrix)}. Such a matrix $A$ is called {\em 
reciprocal} if $a_{ij} = \frac{1}{a_{ji}}$ for every $i,j=1, \ldots ,n$ (then evidently $a_{ii}=1$ for every $i=1, \ldots ,n $). 
A PC-matrix $A$ is called {\em consistent} (or {\em transitive}) if: 
$$a_{ij} \cdot a_{jk}=a_{ik}$$ for every $i,j,k=1,2, \ldots ,n$. For different $i,j,k$ we define {\em a triad} as \linebreak $(a_{ij},a_{ik},a_{jk})$. The triad inconsistency, introduced in \cite{KB1939}, was analyzed only in the cardinal way (by counting intransitive triads).

\cite{KS2015a} shows that it is possible to avoid inconsistency by minimizing the input but such reduction is not always advisable. For example, a multiple choice exam with 100 questions could be reduced to one question with the full mark if answered and zero otherwise, making it more of a lottery that the measure of knowledge. It is pointed out in \cite{KS2015a} that the inconsistency is the result of the excessive input but minimizing the input data to its extreme is not always advisable as we demonstrated by the above example.

\section{Summary}

This study provides a theory, with roofs, for a triad inconsistency axiomatization in pairwise comparisons informally proposed in \cite{KS2014a}. Key properties of the distance-based indicator have been illustrated by several examples. It is of considerable importance for the future research and extensions.

\end{document}